\documentclass{article}
\usepackage[preprint,nonatbib]{neurips_2024}

\usepackage[utf8]{inputenc}
\usepackage[T1]{fontenc}
\usepackage{hyperref}
\usepackage{url}
\usepackage{booktabs}
\usepackage{amsfonts}
\usepackage{nicefrac}
\usepackage{microtype}
\usepackage{xcolor}
\usepackage{csquotes}

\usepackage{amsmath}
\usepackage{amssymb}
\usepackage{mathtools}
\usepackage{amsthm}
\usepackage{subfigure}

\usepackage{mdframed}
\usepackage{wrapfig}
\newcommand{\erclogowrapped}[1]{%
\setlength\intextsep{0pt}%
\begin{wrapfigure}[3]{r}{#1*\real{1.1}}%
\includegraphics[width=#1]{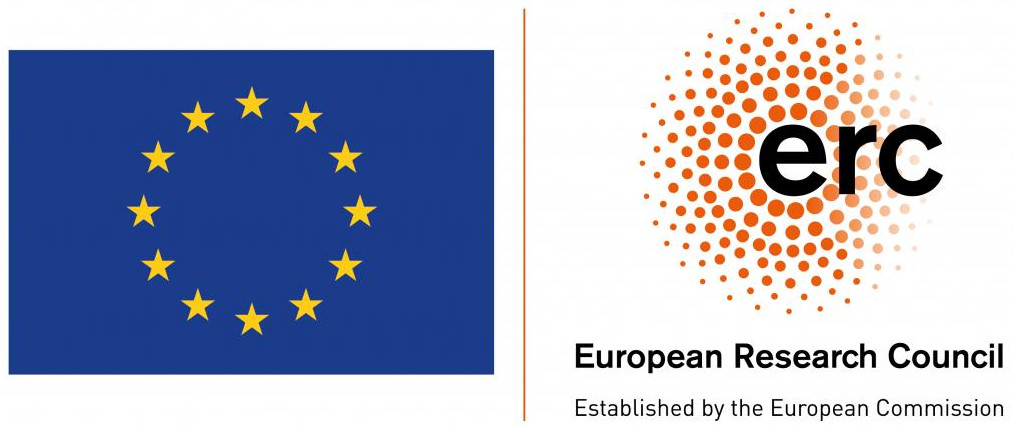}%
\end{wrapfigure}%
}

\newtheorem{theorem}{Theorem}[section]

\newtheorem{lemma}[theorem]{Lemma}
\newtheorem{corollary}[theorem]{Corollary}
\newtheorem{definition}[theorem]{Definition}

\newtheorem{fact}[theorem]{Fact}

\usepackage{hyperref}
\usepackage{cleveref}
\usepackage{comment}
\usepackage[normalem]{ulem}
\usepackage{multirow}

\usepackage{url}

\newcommand{\Alg}{\mathcal{A}}
\newcommand{\Lap}{\mathrm{Lap}}
\newcommand{\densLap}[1]{f_{\mathrm{Lap}}(#1)}
\newcommand{\range}{\mathrm{range}}

\newcommand{\randvar}{Z}
\newcommand{\support}{\mathrm{support}}

\newcommand{\paren}[1]{\left\{ #1 \right\}}

\newcommand{\citet}{\cite}
\newcommand{\citep}{\cite}

\RequirePackage{algorithm}
\RequirePackage{algorithmic}

\renewcommand{\epsilon}{\varepsilon}

\title{Continual Counting with Gradual Privacy Expiration}

\author{%
  Joel Daniel Andersson\\
  Basic Algorithms Research Copenhagen\\
  University of Copenhagen\\
  \texttt{jda@di.ku.dk}\\
  \And
  Monika Henzinger\\
  Institute of Science and Technology, Austria\\
  \texttt{monika.henzinger@ist.ac.at}\\
  \And
  Rasmus Pagh\\
  Basic Algorithms Research Copenhagen\\
  University of Copenhagen\\
  \texttt{pagh@di.ku.dk}\\
  \And
  Teresa Anna Steiner\\
  Technical University of Denmark\\
  \texttt{terst@dtu.dk}\\
  \And
  Jalaj Upadhyay\\
  Rutgers University\\
  \texttt{jalaj.upadhyay@rutgers.edu} \\
}

\begin{document}

\maketitle

\begin{abstract}
Differential privacy with gradual expiration models the setting where data items arrive in a stream and at a given time $t$ the privacy loss guaranteed for a data item seen at time $(t-d)$ is $\epsilon g(d)$, where $g$ is a monotonically non-decreasing function. We study the fundamental {\em continual (binary) counting} problem where each data item consists of a bit, and the algorithm needs to output at each time step the sum of all the bits streamed so far. For a stream of length $T$ and privacy {\em without} expiration continual counting is possible with maximum  (over all time steps) additive error $O(\log^2(T)/\varepsilon)$ and the best known lower bound is $\Omega(\log(T)/\varepsilon)$; closing this gap is a challenging open problem. 

We show that the situation is very different for privacy with gradual expiration by giving upper and lower bounds for a large set of expiration functions $g$. Specifically, our algorithm achieves an additive error of $  O(\log(T)/\epsilon)$ for a large set of privacy expiration functions. We also give a lower bound that shows that if $C$ is the additive error of any $\epsilon$-DP algorithm for this problem, then the product of $C$ and the privacy expiration function after $2C$ steps must be $\Omega(\log(T)/\epsilon)$. Our algorithm matches this lower bound as its additive error is $O(\log(T)/\epsilon)$, even when $g(2C) = O(1)$.

Our empirical evaluation shows that we achieve a slowly growing privacy loss with significantly smaller empirical privacy loss for large values of~$d$ than a natural baseline algorithm.

\end{abstract}

\section{Introduction}\label{sec:intro}
{\em Differential privacy under continual observation}~\cite{chan2011private,dwork2010differentially} has seen a renewed interest recently~\cite{andersson2023smooth,henzinger2022constant,henzinger2023almost,henzinger2024unifying,honaker2015efficient} due to its application in private learning~\cite{choquette2023correlated,choquette2022multi,denisov2022improved,mcmahan2022federated} and statistics~\cite{cardoso2021differentially,chan2012differentially,epasto2023differentially,henzinger2023differentially,kairouz2021practical,smith2017interaction,upadhyay2019sublinear,upadhyay2021differentially,upadhyay2021framework}. In this model, the curator gets the database in the form of a stream and is required to output a given statistic continually. Chan et al.~\cite{chan2011private} and Dwork et al.~\cite{dwork2010differentially} introduced the {\em binary (tree) mechanism} which allows us to estimate the running count of a binary stream of length $T$ with additive error $O(\log^2(T)/\epsilon)$ under $\epsilon$-differential privacy.

The traditional definition of continual observation considers every single entry in the stream equally important for analysis and has equal confidentiality. However, in many applications of continual observation, the data becomes less sensitive with time. For example, consider the case where the stream tracks visit a certain location or website: it being visited a minute ago may constitute more sensitive information than if it was visited a week ago. To capture such scenarios, Bolot et al.~\cite{bolot2013private} defined {\em privacy with expiration}, where the privacy of a streamed data item decreases as a function of the time elapsed since it was streamed. However, known algorithms for privacy with expiration work only in the setting when we expect no privacy after a certain time has elapsed~\cite[Section 6]{bolot2013private}.

This lack of algorithms for privacy with expiration influences some real-world design choices~\cite{vassilvitskii2024personal}. In particular, real-world deployments allocate every user a ``privacy budget'' that is diminished every time their data is used, such that their data should not be used once the privacy budget reaches zero. However, since the data can still be useful, some of these deployments use the heuristic of ``refreshing the privacy budget'', i.e., the privacy budget is reset to a positive default value after a prescribed time period, irrespective of how much privacy budget has been used so far.
This, for example, was pointed out by Tang et al.~\cite{tang2017privacy} in Apple's first large-scale deployment. 
However, refreshing the privacy budget is very problematic as the privacy loss is, in the worst case, multiplied by the number of refreshes, for example, if the old data is reused.

In this paper, we study continual counting with gradual privacy expiration,
generalizing the result in Bolot et al.~\cite{bolot2013private}. Our main contributions are algorithms with the following assets:
\begin{itemize}
    \item \emph{Improve accuracy.}
    We achieve an additive error of $O(\log(T)/\epsilon)$ for a large class of privacy expiration functions and show that this is optimal in a particular sense.
    This is in contrast to continual counting without expiration, where there is a gap of a $\log T$ factor~\cite{dwork2010differentially}. Our work generalizes the $\Omega(\log(T)/\epsilon)$ lower bound for continual counting to a wide class of privacy expiration functions and shows that for any additive error $C$, the product of $C$ and the privacy expiration function after $2C$ steps must be $\Omega(\log T)$.
    We match this lower bound as our additive error is $O(\log(T)/\epsilon)$, even when the expiration function after $2C$ steps is a constant.
    \item \emph{Scale well.}
    Our algorithms work for unbounded streams, run in amortized $O(1)$ time per update, $\log(T)$ space, and offer different trade-offs than conventional continual counting algorithms.
    In allowing for a growing privacy loss, we show that polylogarithmic privacy expiration is sufficient for optimal additive error, and parameterize the algorithm by the speed of the privacy expiration;
    as expected, faster privacy expiration yields a smaller error.
\end{itemize}
We supplement these theoretical guarantees with empirical evaluations. 

\textbf{Related Works.} %
Before presenting our contributions in detail, we give a brief overview of the most relevant related work.
Since Chan et al. \citet{chan2011private} and Dwork et al. \citet{dwork2010differentially}, several algorithms have been proposed for privately estimating {\em prefix-sum under continual observation}, i.e., given a stream of inputs $x_1, x_2, \dots$ from some domain $\mathcal X$, output $y_t = \sum_{i \leq t}x_i$ for all $t \geq 1$. 
{\em Continual binary counting} is a special case of prefix sum when $\mathcal X=\{0,1\}$ and $x_t$ is provided at time $t$. 

When the input is given as a stream, earlier works improved on the basic binary mechanism under (i) distributional assumptions on data~\cite{perrier}, (ii) structural assumptions on data~\cite{rastogi2010differentially}, and (iii) that the importance of data (both with respect to utility and sensitivity) decreases with elapsed time~\cite{bolot2013private}, or (iv) by enforcing certain conditions on the behavior of the output~\cite{chen2017pegasus}. 
In recent work, Fichtenberger et al. \citet{henzinger2022constant} gave algorithms to improve the worst-case non-asymptotic guarantees under continual observation using the {\em matrix mechanism}~\cite{li2010optimizing} and Denisov et al.~\citet{denisov2022improved} used similar approach to provide empirical results that minimize the mean-squared error.  
Subsequently, Henzinger et al.~\citet{henzinger2023almost} showed that the algorithm in Fichtenberger et al.~\citet{henzinger2022constant} achieves almost optimal mean-squared error.

These earlier works are in the traditional definition of privacy under continual observation, i.e., they consider data privacy to be constant throughout the stream. The only exception is the work of Bolot et al.~\citet{bolot2013private}, which defined differentially private continual release with privacy expiration parameterized by a monotonically non-decreasing function $g$ and gave an algorithm for the special case that the data loses all its confidentially after a prescribed time. Our work is in this privacy model. 
There is another line of work motivated by applications in private learning that studies privacy-preserving prefix sum without restricting access to the data points (such as allowing multiple passes)~\cite{choquette2022multi, koloskova2023convergence} and providing privacy-preserving estimates under various privacy models like shuffling~\cite{choquette2023amplified}. 
Since we focus on continual observation, we do not compare our results with this line of work. 

The work whose techniques are the most related to ours is the algorithm in Dwork et al. \cite{dwork2010differentially} for continual counting satisfying pan-privacy~\cite{dwork2010pan}.
Roughly speaking, an algorithm is {\em pan-private} if it is resilient against intruders who can observe snapshots of the internal states of the algorithms. %

\subsection{Our Contributions}
\label{sec:contribution}
We start by first formally stating the problem.
As mentioned above, the focus of this work is privacy with expiration given as Definition 3 in  Bolot et al.\citet{bolot2013private}:
\begin{definition}\label{def:decayed} 
    Let $g:\mathbb{N} \rightarrow \mathbb{R}_{\geq 0}$ be a non-decreasing function\footnote{We believe $g$ being {\em non-increasing} is a typo in Bolot et al.~\citet{bolot2013private}.}. %
    Let $\mathcal{A}$ be a randomized online algorithm that takes as an input a stream $x_1, x_2, \dots$ and at every time step $t$ outputs $\mathcal{A}(x_1,\dots,x_{t})$.
    $\mathcal{A}$ satisfies \emph{$\epsilon$-differential privacy with expiration (function) $g$} if for all $\tau \geq 1$, for all measurable $S\subseteq\range(\Alg)^{*}$ 
    \footnote{``$*$" is the Kleene operator: given a set $\Sigma$, $\Sigma^*$ is the (infinite) set of all finite sequences of elements in $\Sigma$.}, all possible inputs $x_1,\dots, x_\tau$, all $j\leq \tau$ and all $x'_j$ with $|x_j-x'_j|\in[0,1]$ 
    \begin{align*}
        \Pr[(\mathcal{A}(x_1,\dots,x_j,\dots,x_t))_{t=1}^{\tau}\in S] & \quad \leq e^{g(\tau-j)\epsilon}\Pr[(\mathcal{A}(x_1,\dots,x'_j,\dots,x_t))_{t=1}^{\tau}\in S],
    \end{align*}
    where the probability is over the coin throws of $\mathcal{A}$. We refer to $g(\tau-j)\cdot \epsilon$ as the privacy loss.
\end{definition}

Letting $T$ be the length of the input stream, the best known bound on the $\ell_{\infty}$-error for continual counting under $\epsilon$-differential privacy is $O(\log^2(T)/\epsilon)$, achieved by the algorithms in \cite{dwork2010differentially,chan2011private}. Alternatively, the analysis of \cite{chan2011private} can be used to show that running this algorithm with $\epsilon'=\epsilon\log(T)$ achieves privacy loss with expiration function $g(d)=\log(T)$ for all $d=1,\dots, T$,
and error $O(\log(T)/\epsilon)$. Our main contribution is to show that better trade-offs are possible: In particular, we can achieve the same error with a \emph{strictly smaller function $g$}, i.e. we can get an $O(\log(T)/\epsilon)$ bound on the $\ell_{\infty}$-error with an expiration function of $g(d)\approx \log d$. %
More generally, our algorithm provides a trade-off between privacy loss and both $\ell_{\infty}$-error and expected $\ell_2^2$-error for all expiration functions $f(d)$ that satisfy  (roughly) $f(d)\geq 1 + \log^\lambda(d)$ for any $\lambda>0$. 
The exact expiration function $g$ is stated below in \Cref{thm:detailed_main_upper}.
It also includes a parameter $B$ that allows the privacy loss to be ``shifted'' by $B$ time steps, i.e., there is no privacy loss in the first $B$ time steps. If the length $T$ of the stream is unknown, then $B$ is a constant. If $T$ is given to the algorithm, then $B$ can be a function of $T$. 

By definition~\ref{def:decayed}, any algorithm satisfying differential privacy with expiration $g$ also fulfills differential privacy with any expiration function that is pointwise at least as large as $g$. Specifically, for two functions $f$ and $g$ defined on the same domain $\mathcal{D}$, we say $f\succeq g$ if $f(x)\geq g(x)$ for all $x\in \mathcal{D}$. We are now ready to state our main theorem:%

\begin{theorem}\label{thm:detailed_main_upper}\label{thm:main_upper}
Let $\lambda\in\mathbb{R}_{>0}\backslash\{\tfrac{3}{2}\}$ be a constant, and let parameters $\varepsilon \in \mathbb{R}_{>0}$ and $B\in\mathbb{N}$ be given.
There exists an algorithm $\mathcal{A}$ that approximates prefix sums of a (potentially unbounded) input sequence $x_1, x_2, \dots$ with $x_i \in [0,1]$ satisfying $\varepsilon$-differential privacy with any expiration function $f$ such that $f \succeq g$, where
\[
    g(d) = \begin{cases}
    0 & \text{for } d < B\\
    O(1 + \log^{\lambda}(d-B+1)) & \text{for } d \geq B
    \end{cases}
\]

Considering all releases up to and including input $t$, the algorithm $\mathcal{A}$ uses $O(B+\log t)$ space and $O(1)$ amortized time per input/output pair and has the following error guarantees at each individual time step $t$ for $\beta > 0$, %
\begin{itemize}
    \item $\mathbb E_{\mathcal A}\left[(\mathcal{A}(x) - \sum_{i=1}^t x_i)^2\right] = O(B^2+\log^{3-2\lambda}(t)/\varepsilon^2)$,
    \item $|\mathcal{A}(x) - \sum_{i=1}^t x_i| = O(B + \log^q(t) \sqrt{\log(1/\beta)}/ \varepsilon)$ with probability $1-\beta$ where $q=\max(1/2, 3/2 - \lambda)$.
\end{itemize}
\end{theorem}

The case when $\lambda\in\{0,3/2\}$ is covered in \Cref{sec:missingproof}. Note that choosing $\lambda > 3/2$ implies a constant expected squared error at each time step if $B=O(1)$.
Parameter $\lambda$ controls the trade-off between the \emph{asymptotic} growth of the expiration function and the error, while
$\varepsilon$ controls the trade-off between \emph{initial} privacy (after $B$ time steps which is $\epsilon \cdot g(B)$ ) and the error, which is inversely proportional to $\epsilon$.
Also, for releasing $T$ outputs we have the following corollary.
\begin{corollary}\label{cor:main}
    The algorithm $\mathcal A$ with $B=O(\log(T)/\epsilon)$ and $\lambda\geq 1$ incurs a maximum (over all time steps) additive  $\ell_{\infty}$-error of $O(\log(T)/\epsilon)$ when releasing $T$ outputs with probability $1-1/T^c$, for constant $c > 0$, and achieves privacy with expiration function $g$ as in Theorem~\ref{thm:main_upper}. 
\end{corollary}
   
In \Cref{sec:empirical}, we provide empirical evidence to show that we achieve a significantly smaller empirical privacy loss than a natural baseline algorithm.
Finally, we complement our upper bound with the following lower bound shown in Appendix~\ref{sec:omitted_lowerbound}.
\begin{theorem}
\label{thm:main_lower}
      Let $\Alg$ be an algorithm for binary counting for streams of length $T$ which satisfies $\epsilon$-differential privacy with expiration $h$.
      Let $C$ be an integer such that $\Alg$ incurs a maximum additive error of at most $C < T/2$ over $T$ time steps with a probability of at least $2/3$. Then
    \begin{align*}
        2C\cdot h(2C-1) \geq \frac{\log(T/(6C))}{2\epsilon}\enspace .
    \end{align*}
    \vspace{-3mm}
\end{theorem}
Note that \Cref{thm:main_lower} gives a lower bound for $h(j)$ for a specific $j$, namely $j = 2C-1$, and as $h$ is non-decreasing by Definition \ref{def:decayed}, the lower bound also holds for all $h(j')$ with $j'\geq j$. 

Note that \Cref{thm:main_lower} shows that our algorithm in Corollary~\ref{cor:main} achieves a tight error bound for the expiration functions  $\lambda \ge 1$ and $B =O(\log (T)/\epsilon)$. Assume $\mathcal A' $ is an algorithm that approximates prefix sums in the continual setting and which satisfies differential privacy with expiration function $h$ and maximum error $C\leq B/2+1$ at all time steps with probability at least 2/3 for an even $B$. When run on a binary input sequence, $\mathcal A'$ solves the binary counting problem. Thus,
by \Cref{thm:main_lower}, we have that $2C \cdot h(B+1)\geq 2C\cdot h(2C-1) \geq \frac{\log(T/(6C))}{2\epsilon} \ge \frac{\log(T/(3B+6))}{2\epsilon} = \Omega(\frac{\log T}{\epsilon})$. 

Now consider the algorithm $\mathcal A$ given in Corollary~\ref{cor:main} and note that, by definition of the expiration function $g$, $g(B+1) = O(1)$ and that Corollary~\ref{cor:main} shows that $C = O(\log(T)/\epsilon)$. This is tight as \Cref{thm:main_lower} shows that for such an expiration function $C = \Omega(\log(T)/\epsilon)$.

\subsection{Technical Overview}
Central to our work is the event-level pan-private algorithm for continual counting by Dwork et al.~\citet{dwork2010differentially}.
Similarly to the binary tree algorithm of Dwork et al.~\citet{chan2011private}, a noise variable $z_I$ is assigned to every {\em dyadic interval} $I$ (see Appendix~\ref{sec:log-decay} for a formal definition) contained in $[0, T-1]$. Let this set of dyadic intervals be called $\mathcal{I}$.
In the version of the binary tree algorithm of Chan et al.~\citet{chan2011private}, the noise added to the sum of the values so far (i.e.~the non-private output) at any time step $1\leq t\leq T$ is equal to $\sum_{I\in D_{[0, t-1]}} z_I$, where $D_{[0, t-1]} \subseteq \mathcal{I}$ is the dyadic interval decomposition of the interval $[0, t-1]$.
The pan-private algorithm adds different noise to the output: it adds at time $t$ the \emph{noises for all intervals containing $t-1$}, i.e.,~$\sum_{I\in\{I\in\mathcal{I}\,:\, t-1\in I\}} z_I$.
This pan-private way of adding noise helps us bound the privacy loss under expiration. %
For two neighboring streams differing at time step~$j$, we can get the same output at $\tau \geq j$ by shifting the values of the noises of a set of disjoint intervals covering $[j,\tau]$ each by at most 1.
Using that $|D_{[j, \tau]}| = O(\log(\tau - j + 1))$, we show that the algorithm satisfies a logarithmic privacy expiration.

In our algorithm, we make four changes to the above construction (i.e., the pan-private construction in Dwork et al. \citet{dwork2010differentially}):
\textbf{(i)} We do not initialize the counter with noise separately from that introduced by the intervals.
\textbf{(ii)} We split the privacy budget unevenly across the levels of the dyadic interval set instead of uniformly allocating it.
This allows us to control the asymptotic growth of the expiration function, and the error. %
This change, however, requires a more careful privacy analysis. 
\textbf{(iii)} At time $t$ we add noise identified by intervals containing $t$, not $t-1$.
While this is a subtle difference, it allows us to exclude intervals starting at $0$ from $\mathcal{I}$, leading to our algorithm running on \emph{unbounded} streams with utility that depends on the current time step $t$.
Said differently, our algorithm does not need to know the stream length in advance. 
This is in contrast to Dwork et al.~\citep{dwork2010differentially} where the construction requires an upper bound $T$ on the length of the stream so that the utility guarantee at each step is fixed and a function of $T$.
\textbf{(iv)} We allow for a delay of $B$, meaning we output $0$ for the initial $B$ steps.
This gives perfect privacy for the first $B$ steps, and, since each element of the stream is in $[0,1]$, the delayed start leads only to an additive error of $O(B)$.

\section{Preliminaries}\label{sec:prelims}
Let $\mathbb N_{>0}$ denote the set $\{1,2,\cdots\}$ and $\mathbb R_{\geq 0}$ the set of non-negative real numbers. We use the symbol $g$ to denote the function that defines the privacy expiration, i.e., $g:\mathbb N \to \mathbb R_{\geq 0}$. We fix the symbol $\mathcal A(x)$ to denote the randomized streaming algorithm that, given an input $x=x_1, x_2, \dots$ as a stream, provides $\epsilon$-differential privacy with expiration $g$.
All algorithms {\em lazily draw noise}, meaning that a ``noise'' random variable is only drawn when first used and is re-used and {\em not} re-drawn when referenced again.
For a random variable, $Z$, we use $\mathsf{support}(Z)$ to denote its support. For a sequence of random variables $Z=\randvar_1,\randvar_2,\dots$, we use $\support(Z)$ to denote $\support(\randvar_1)\times\support(\randvar_2)\times \dots$.

\textbf{Helpful lemmas.}
We now collect some helpful lemmas shown formally in Appendix~\ref{sec:omitted_proofs}.
To show privacy with expiration in the following sections, we repeatedly use the following observation: a similar lemma has been used to show the standard definition of differential privacy, e.g., in the proof of Theorem 2.1 of Dwork et al.  \citet{dwork2015pure}. Informally, it says that if there exists a map $q$ between random choices made by algorithm $\Alg$ such that for any input $x$ and fixed sequence $z$ of random choices, the map returns a sequence $q(z)$ such that (1) the output $\Alg(x, z)$ equals the output $\Alg(x', q(z))$ and (2) the probability of picking $z$ is similar to the probability of picking $q(z)$, then $\Alg$ is private.
The notion of ``similar probability'' is adapted to the definition of differential privacy with expiration and depends on the function $g$. All the results in this section are shown formally in  \Cref{sec:omitted_proofs_prelims}: %
\begin{fact}\label{fact:varialbe_shift}
     Consider an algorithm   $\Alg$ that uses a sequence of random variables $Z=\randvar_1,\randvar_2,\dots$ as the only source of randomness.
     We can model $\Alg:\chi\times\support(Z)$ as a (deterministic) function of its actual input from the universe $\chi$ and the sequence of its random variables $Z$. Suppose that for all $\tau\in\mathbb{N}_{>0}$, $j\leq \tau$ and all neighboring pairs of input streams $x=x_1,\dots,x_j,\dots,x_\tau$ and $x'=x_1,\dots,x'_j,\dots,x_\tau$, there exists a  function
     $q:\support(Z)\rightarrow \support(Z)$ such that $\Alg(x; z)=\Alg(x'; q(z))$ and 
     $$\Pr_{z\sim Z}[z\in\mathcal{N}]\leq e^{\epsilon g(\tau-j)}\, \Pr_{z'\sim Z}[z'\in q(\mathcal{N})] \quad \text{ for all } \quad \mathcal{N}\subseteq\support(Z).$$
    Then $\Alg$ satisfies $\epsilon$-differential privacy with expiration $g$.
\end{fact}
\begin{lemma}\label{lem:lap_shift}
    Let $Z=\randvar_1,\randvar_2,\dots,\randvar_k$ be a sequence of independent Laplace random variables, such that $\randvar_i \sim \Lap(b_i)$ for $b_i>0$, for all $i\in[k]$. Let $q$ be a bijection $q:\support(Z)\rightarrow\support(Z)$ of the following form: %
    For all $\Delta:=\begin{pmatrix} \Delta_1 & \Delta_2 & \cdots & \Delta_k \end{pmatrix}\in\mathbb{R}^k$, and for all $z\in\support(Z)$, we have $q(z)=z+\Delta\in\support(Z)$. %
    Then  for all $\mathcal{N}\subseteq\support(Z)$ we have
    \begin{align*}
        \Pr_{z\in Z}[z\in\mathcal{N}]\leq e^s \, \Pr_{z\in Z}[z\in q(\mathcal{N})], \quad \text{where} \quad s = \sum_{i=1}^k {|\Delta_i|\over b_i}.
    \end{align*}
\end{lemma} 

\section{Warmup}\label{sec:warmup}
As a warm-up, we give two simple algorithms for two obvious choices of the expiration function: the linear expiration function $g(d)=d$ and the logarithmic expiration function $g(d)=2\log(d+1)+2$. %

\subsection{A Simple Algorithm with Linear Privacy Expiration}%
First, we consider a simple algorithm which gives $\epsilon$-differential privacy with expiration $g:\mathbb N \to \mathbb R_{\geq 0}$, where $g(d)=d$. The maximum error of this algorithm over $T$ time steps is bounded by $O(\epsilon^{-1}\log(T/\beta))$, with probability at least $1-\beta$. The algorithm $\Alg_{\mathsf{simple}}$ is given in Algorithm~\ref{alg:simple_privacy_degradation}.
It adds fresh Laplace noise to any output sum. Note that this is the same algorithm as the Simple Counting Mechanism I from Chan et al.~\citet{chan2011private}.
However, we show that for the weaker notion of differential privacy with linear expiration, Laplace noise with \emph{constant} scale suffices, even though the sensitivity of $\Alg_{\mathsf{simple}}$ running on a stream of length $T$ is $T$.
To prove this, we show that for two neighbouring streams differing at time step $j$, we obtain the same output by ``shifting'' the values of the Laplace noises for all outputs after step $j$ by at most 1. 
We defer the proof of the following lemma to Appendix~\ref{sec:omitted_warm_up_proof}.
\begin{lemma}\label{lem:simple}
    The algorithm $\Alg_{\mathsf{simple}}$, given in Algorithm~\ref{alg:simple_privacy_degradation}, is $\epsilon$-differentially private with expiration $g$, where $g:\mathbb{N}\rightarrow \mathbb{N}$ is the identity function $g(d)=d$ for all $d\in\mathbb{N}$.
    It incurs a maximum additive error of $O(\epsilon^{-1}\log (T/\beta))$ over all $T$ time steps simultaneously with probability at least $1-\beta$.
\end{lemma}

%

\newsavebox{\algleft}
\newsavebox{\algright}

\savebox{\algleft}{%
\begin{minipage}{.47\textwidth}
\begin{algorithm}[H]
\begin{algorithmic}[1]
    \STATE {\bf Input:} A stream $x_1, x_2, \cdots$, privacy parameter $\epsilon$
    
   \STATE Lazily draw 
   $\randvar_{t-1}\sim\Lap\left(\frac{1}{\epsilon}\right)$
   \STATE At time $t=1$, output $0$\;
    
    \FOR{$t=2$ {\bf to} $\infty$}
   \STATE At time $t$, output $\sum_{i=1}^{t-1}x_i+\randvar_{t-1}$\ENDFOR
    \end{algorithmic}
    \caption{$\Alg_{\mathsf{simple}}$: Continual counting under linear gradual privacy expiration}
\label{alg:simple_privacy_degradation}
\end{algorithm}
\end{minipage}}
\savebox{\algright}{%
\begin{minipage}{.51\textwidth}
\begin{algorithm}[H]
\begin{algorithmic}[1]
    \STATE {\bf Input:} A stream $x_1, x_2,\cdots$, privacy param $\epsilon$
    \STATE Let $\mathcal{I}$ be the dyadic interval set on $[1,\infty)$ \;
    
    \STATE For each $I\in\mathcal{I}$, draw i.i.d.\ %
    $\randvar_I\sim\Lap\left({1\over \epsilon } \right)$\;
    
   \STATE Set $\mathcal{I}_t=\{I\in\mathcal{I}:t \in I\}$\;
    
    \FOR{$t=1$ {\bf to} $\infty$}
    \STATE At time $t$, output $\sum_{i=1}^{t} x_i+\sum_{I\in\mathcal{I}_t}\randvar_I$\;
    \ENDFOR
    \end{algorithmic}
    \caption{$\Alg_{\mathsf{log}}$: Continual counting under logarithmic gradual privacy expiration}
\label{alg:simpe_log_privacy_degradation}

\end{algorithm}
\end{minipage}}

\noindent\usebox{\algleft}\hfill\usebox{\algright}%

\subsection{A Binary-Tree-Based Algorithm with Logarithmic Privacy Expiration}\label{sec:log-decay}

Next, we show how an algorithm similar to the binary tree algorithm~\cite{dwork2010differentially} gives $\epsilon$-differential privacy with expiration $g:\mathbb N \to \mathbb R_{\geq 0}$, where $g(d)=2\log(d+1)+2$. This result can also be derived from \Cref{thm:main_upper} by setting $\lambda=1$ and $B=0$. 
As in the case when $g(d)=d$, the maximum error of this algorithm over $T$ time steps is again bounded by $O(\epsilon^{-1}\log(T/\beta))$, with probability at least $1-\beta$.
Similarly to the binary tree algorithm, we define a noise variable for every node in the tree, but we do this in the terminology of \emph{dyadic intervals}.
We consider the \emph{dyadic interval set} $\mathcal{I}$ on $[1, \infty)$ (formally defined shortly), associate a noise variable $z_I$ with each interval $I\in\mathcal{I}$, and at time step $t$ add noise $z_I$ for each $I\in\mathcal{I}$ that contains $t$. 
This is similar to the construction in Dwork et al.~\citet{dwork2010differentially}, with the exception that 
they instead consider the dyadic interval set on $[0, T-1]$, add noise $z_I$ at time $t$ if $t-1\in I$, and initialize their counter with noise from the same distribution.
Our choice of $\mathcal{I}$ allows the algorithm to run on unbounded streams, and leads to adding up $1 + \lfloor\log(t)\rfloor$  noise terms at step $t$ rather than $1 + \lfloor\log(T)\rfloor$.
For privacy, we will argue that if two streams differ at time $j$, then we get the same outputs up to time $\tau\geq j$ by considering a subset of disjoint intervals in $\mathcal{I}$ covering $[j, \tau]$, and shifting the associated Laplace random variables appropriately.
In the following, we describe this idea in detail. We start by describing the dyadic interval decomposition of an interval.

\textbf{Dyadic interval decomposition.}
For every non-negative integer $\ell$, we divide $[1, \infty)$ into disjoint intervals of length~$2^{\ell}$:
$
    \mathcal{I}^{\ell}=\{[k\cdot 2^{\ell},(k+1)\cdot 2^{\ell} - 1], k\in\mathbb{N}_{>0}\}.
$ 
We call $\mathcal{I}=\bigcup_{\ell=0}^{\infty}\mathcal{I}^{\ell}$ the \emph{dyadic interval set} on $[1, \infty)$, and $\mathcal{I}^\ell$ the $\ell$-th \emph{level} of the dyadic interval set.
We show the following in Appendix~\ref{sec:omitted_proofs}.
\begin{fact}\label{fact:dyadic_decomposition}
    Let $\mathcal{I}$ be the dyadic interval set on $[1,\infty)$. For any interval $[a,b]$, $1\leq a \leq b$, there exists a set of intervals $D_{[a,b]}\subseteq \mathcal{I}$, referred to as the \emph{dyadic interval decomposition} of $[a,b]$, such that (i) the sets in $D_{[a,b]}$ are disjoint; (ii) $\bigcup_{I\in D_{[a,b]}}I=[a,b]$; and (iii) $D_{[a,b]}$ contains at most 2 intervals per level, and the highest level $\ell$ of an interval satisfies $\ell\leq \log(b-a+1)$ 
\end{fact}
\begin{fact}\label{fact:dyadic_levels}
    Let $\mathcal{I}$ be the dyadic interval set on $[1,\infty)$, and for $t\in\mathbb{N}_{>0}$ define $\mathcal{I}_t = \{ I\in\mathcal{I} : t\in I\}$ as the \emph{intersection} of $t$ with $\mathcal{I}$. Then $|\mathcal{I}_t| = \lfloor \log t \rfloor + 1$.
\end{fact}

\begin{lemma}
    The algorithm $\Alg_{\mathsf{log}}$ given in Algorithm~\ref{alg:simpe_log_privacy_degradation} satisfies $\epsilon$-differential privacy with expiration $g$, where $g:\mathbb{N}\rightarrow \mathbb{R}$ is defined as $g(x)=2\log(x+1)+2$.
    It incurs a maximum additive error of $O(\epsilon^{-1}\log (T/\beta))$ over all $T$ time steps simultaneously with probability at least $1-\beta$.
\end{lemma}

%

\textbf{Privacy.}
We use Fact~\ref{fact:varialbe_shift} and Lemma~\ref{lem:lap_shift} to argue privacy of $\Alg_{\mathsf{log}}$: Let $x$ and $x'$ differ at time $j$. Note that the prefix sums fulfill the following properties:
(i) $\sum_{i=0}^t x_i=\sum_{i=0}^t x'_i$ for all $t<j$ and (ii) $\sum_{i=0}^t x'_i = \sum_{i=0}^t x_i+y$ for all $t\geq j$, where $y=x'_j-x'_j \in [-1,1]$.

In the following, we refer to the output of the algorithm run on input $x$ and with values of the random variables $z$ as $\Alg_{\mathsf{log}}(x;z)$.
Let $\tau \geq j$ be given, and consider $S\subseteq\range(\Alg_{\mathsf{log}})^{*}$. %
Let $Z=(\randvar_I)_{I\in\mathcal{I}}$ be the sequence of Laplace random variables used by the algorithm. For any fixed output sequence $s\in S$, let $z=(z_I)_{I\in\mathcal{I}}$ be a sequence of values that the Laplace random variables need to assume to get output sequence $s$ for input $x$. That is $\sum_{i=1}^t x_i + \sum_{I\in\mathcal{I}_t} z_I=s_t$ for $t\geq 1$. Let $D_{[j,\tau]}$ be the decomposition of $[j,\tau]$ as defined in Fact~\ref{fact:dyadic_decomposition}.  
We define a bijection $q$ satisfying the properties of Fact~\ref{fact:varialbe_shift} as follows: $q(z)=z'=(z'_I)_{I\in\mathcal{I}}$ such that 
\begin{align*}
    z'_I = 
    \begin{cases}
    z'_I=z_I &\forall I\notin D_{[j,\tau]} \\
    z'_I=z_I+y &\forall I\in D_{[j,\tau]}
    \end{cases}. 
\end{align*} We show the two properties needed to apply Fact~\ref{fact:varialbe_shift}:

    (1) Note that $(\Alg_{\mathsf{log}}(x;z))_t=\sum_{i=1}^t x_i+\sum_{I\in \mathcal{I}_t} z_I$.
    For $t< j$, we have $t\notin [j,\tau]$  and therefore $t\notin I$ for any $I\in D_{[j,\tau]}$. Therefore, we have \begin{align*}
        \sum_{i=1}^t x_i+\sum_{I\in \mathcal{I}_t} z_I&=\sum_{i=1}^t x'_i+\sum_{I\in \mathcal{I}_t} z_I=\sum_{i=1}^t x'_i+\sum_{I\in \mathcal{I}_t} z'_I.
    \end{align*} For $j\leq t\leq \tau$, we have that $t$ is contained in exactly one $I\in D_{[j,\tau]}$. %
    Thus, $z'_I=z_I+y$ for exactly one $I\in\mathcal{I}_t$, and $z'_{I'}=z_{I'}$ for all $I'\in \mathcal{I}_t\backslash\{I\}$. 
    Further, since $t\geq j$, we have that $\sum_{i=1}^t x_i=\sum_{i=1}^t x'_i-y$. 
    Together, this shows the first property of Fact~\ref{fact:varialbe_shift} as
    \begin{align*}
        \sum_{i=1}^t x_i+\sum_{I\in \mathcal{I}_k} z_I &=\sum_{i=1}^t x'_i-y+\sum_{I\in \mathcal{I}_t} z'_I+y =\sum_{i=1}^t x'_i+\sum_{I\in \mathcal{I}_t} z'_I,
    \end{align*} 

    (2) By Fact~\ref{fact:dyadic_decomposition}, $|D_{[j,\tau]}|\leq 2(\log(\tau-j+1)+1)$. Thus, by Lemma~\ref{lem:lap_shift} for any $\mathcal{N}\in \support(Z)$,
    \begin{align*}
        \Pr_{z\in Z}[z\in\mathcal{N}] &\leq e^{\sum_{I\in D_{[j,\tau]}}|y|\epsilon}\Pr_{z\in Z}[z\in q(\mathcal{N})] \leq e^{2(\log(\tau-j+1)+1)\epsilon}\Pr_{z\in Z}[z\in q(\mathcal{N})],
    \end{align*}
    so the second property of Fact~\ref{fact:varialbe_shift} is fulfilled with $g(x)=2\log(x+1)+2$.
By Fact~\ref{fact:varialbe_shift}, we have differential privacy with privacy expiration $g(x)=2\log(x+1)+2$.%
\paragraph{Accuracy.}
To show accuracy at step $t$, let $Y_t=\sum_{I\in\mathcal{I}_t}\randvar_I$, i.e. the noise added at time step $t$.
By Fact~\ref{fact:dyadic_levels} we add $k=|\mathcal{I}_t| = \lfloor\log t\rfloor + 1$ Laplace noises with scale ${1\over \epsilon}$.
Let $M_{t,\beta}=\max\paren{\sqrt{\lfloor\log t\rfloor + 1},\sqrt{\ln(2/\beta)}}$.
By Corollary~\ref{cor:sum_of_eq_lap}, we have that 
$
    \Pr\left[|Y_t|>{2\over\epsilon}\sqrt{2\ln(2/\beta)}M_{t,\beta}\right]\leq \beta,
$ for any $\beta<1$.
Setting $\beta=\beta'/T$, it follows that with probability at least $1-\beta'$,  $|Y_t|=O(\epsilon^{-1}\log(T/\beta'))$ for all time steps $t \leq T$ simultaneously.
\section{Proof of Theorem~\ref{thm:main_upper} and Corollary~\ref{cor:main}}\label{sec:main_proof}

Section~\ref{sec:log-decay} shows that we can obtain an error smaller than the binary mechanism by using differential privacy with a logarithmic expiration function. Here we show a general trade-off between the expiration function's growth and the error's growth. %
Two techniques are needed to show the theorem:
    
    {\it Delay.} All outputs are delayed for $B$ steps. 
    That is, in the first $B$ steps, the mechanism outputs $0$, thereafter, it outputs a private approximation of $\sum_{i=1}^{j-B} x_i$.
    Delay introduces an extra additive error of up to $B$, but ensures perfect privacy for the first $B$ time steps after receiving an input.

{\it Budgeting across levels.} The privacy budget is split unevenly across levels of the dyadic interval set
    in order to control the asymptotic growth of the expiration function.
    Specifically, the budget at level $\ell$ is chosen to be proportional to $(\ell+1)^{\lambda - 1}$.
    The case $\lambda = 1$ corresponds to the even distribution used in the construction of Section~\ref{sec:log-decay}.
%

\begin{algorithm}[t]
\begin{algorithmic}[1]
    \STATE {\bf Input:} A stream $x_1, x_2\cdots$, privacy parameter $\epsilon$,\\
    parameters $B \in \mathbb N$, $\lambda \in \mathbb R_{>0}\backslash\{3/2\}$

    \STATE Let $\mathcal{I}$ be the dyadic interval set on $[1,\infty)$ and $\mathcal{I}_t=\{I\in\mathcal{I}:t - B\in I\}$ \;
    
    \STATE $\forall\ell$, $I\in\mathcal{I}^\ell$, lazily draw i.i.d.\ %
    $\randvar_I\sim\Lap\left({(1+\ell)^{1-\lambda}\over \epsilon} \right)$ \;

    \FOR{$t=1$ {\bf to} $B$}
    
   \STATE At time $t$, output $0$\;
    \ENDFOR
    \FOR{$t=B+1$ {\bf to} $\infty$}
   \STATE At time $t$, output $s_{t-B} = \sum_{i=1}^{t-B} x_i+\sum_{I\in\mathcal{I}_{t-B}}\randvar_I$\;
    \ENDFOR
    \end{algorithmic}
    \caption{Continual counting, gradual privacy expiration}
\label{alg:privacy_degradation}
\end{algorithm}

\noindent
Our algorithm is shown as Algorithm~\ref{alg:privacy_degradation}.
In the following, we refer to the output of the algorithm run on input $x$ and with values of the random variables $z$ as $\Alg(x;z)$.

\textbf{Privacy.}
We now show that $\Alg$ satisfies Definition~\ref{def:decayed}.
For $x$ and $x'$ that differ (only) at time $j\leq \tau$, the prefix sums fulfill the following properties:
\begin{itemize}
    \item $\sum_{i=0}^t x_i=\sum_{i=0}^t x'_i$ for all $t<j$ and
    \item $\sum_{i=0}^t x'_i = \sum_{i=0}^t x_i+y$ for all $j\leq t \leq \tau$, where $y=x'_j-x_j \in [-1,1]$.
\end{itemize}
Let $j$ and $\tau$ be defined as in Definition~\ref{def:decayed}.
Due to the delay, if $\tau' = \tau - B < j$
(corresponding to $d < B$) the privacy claim is immediate since the output distributions of $\Alg$ up to step $\tau$ are identical on the two inputs. Otherwise, for $\tau' = \tau - B \ge j$, i.e., for $d \geq B$, we wish to use Fact~\ref{fact:varialbe_shift} and Lemma~\ref{lem:lap_shift}.
Let $Z=(\randvar_I)_{I\in\mathcal{I}}$ be the sequence of Laplace random variables used by $\Alg$. 
For input $x$ consider a fixed length-$\tau$ output sequence consisting of $B$ zeros followed by $s_1,\dots,s_{\tau'}$.
Let $z=(z_I)_{I\in\mathcal{I}}$ be a sequence of values for the Laplace random variables in order to produce this output sequence with input $x$. 
That is, $s_t = \sum_{i=1}^t x_i + \sum_{I\in\mathcal{I}_t} z_I$ for $t\geq 1$.
Let $D_{[j,\tau']}$ be the decomposition of $[j,\tau']$ as defined in Fact~\ref{fact:dyadic_decomposition}.  
We define a bijection $q$ satisfying the properties of Fact~\ref{fact:varialbe_shift} as follows: $q(z)=z'=(z'_I)_{I\in\mathcal{I}}$ such that 
\begin{align*}
    z'_I = 
    \begin{cases}
    z'_I=z_I &\forall I\notin D_{[j,\tau']} \\
    z'_I=z_I+y &\forall I\in D_{[j,\tau']}
    \end{cases}. 
\end{align*} 

We show the two properties needed to apply Fact~\ref{fact:varialbe_shift}:
\begin{lemma}\label{lem:mainalgdecay}
    The function $q(z)$ satisfies  $\Alg(x;z)=\Alg(x';q(z))$ and for $g(d)=O(1+\frac{1}{\lambda}\log^{\lambda}(d-B+1))$ we have
         $\Pr_{z\sim Z}[z\in \mathcal{N}]\leq e^{\epsilon g(\tau-j)}\Pr[z'\in \mathcal{N}]$ for all $\mathcal{N}\subseteq \support(Z)$.
\end{lemma}

\textbf{Space and time.}
Algorithm~\ref{alg:privacy_degradation} can update the sums $\sum_{i=1}^{t-B} x_i$ and $\sum_{I\in\mathcal{I}_t}\randvar_I$ in each time step $t$ using the following idea:
The $B$ most recent inputs are kept in a buffer to allow calculation of prefix sums with delay and also the $\lfloor\log(t)\rfloor + 1$ random variables of those values that were added to the most recent output.
At a given step, each random value that is no longer used is subtracted from the most recent output, and each new random value is added.
An amortization argument as in the analysis of the number of bit flips in a binary counter yields the $O(1)$ amortized bound.

\textbf{Accuracy.}
To show the accuracy guarantee, we need to account for the error due to delay as well as the noise required for privacy.
It is easy to see that the delay causes an error of at most $B$, since the sum of any $B$ inputs is bounded by $B$.
Thus, for both error bounds it remains to account for the error due to noise.
At every time step $t$ after $B$, the output is the delayed prefix sum plus a sum of Laplace distributed noise terms as indicated by $\mathcal{I}_{t-B}$ with parameters $b_\ell = (\ell + 1)^{1-\lambda}/\epsilon$, where $\ell = 0,\dots,\lfloor\log(t-B)\rfloor$.
To bound the variance of the noise, $2\sum_\ell b_\ell^2$ we compute: 
\begin{alignat*}{1}
    \sum_{\ell=0}^{\lfloor\log(t-B)\rfloor} b_\ell^2 & = {1 \over \epsilon^2}\sum_{\ell=0}^{\lfloor\log(t-B)\rfloor} (\ell + 1)^{2(1-\lambda)}
     \leq {1 \over \epsilon^2} \left(1 + \int\limits_{1}^{\log(t)+2} x^{2(1-\lambda)} \mathsf dx\right)\\
    & \leq {1 \over \epsilon^2} \left(1 + \bigg[\frac{1}{3 - 2\lambda} x^{3 - 2\lambda}\bigg]_{x=1}^{x=\log(t) + 2} \right)
     = O\left( {1+\log(t)^{3-2\lambda} \over \varepsilon^2} \right) \enspace .
\end{alignat*}

This calculation assumes $\lambda \ne 3/2$, proving the statement on the squared error in Theorem~\ref{thm:detailed_main_upper}.
For the high probability bound we invoke Lemma~\ref{lem:sum_of_lap} with
 $  b_M  = \max_\ell(b_\ell) = \max ( 1, \log(2t)^{1-\lambda}) / \varepsilon ) = O(( 1 + \log(t)^{1-\lambda}) / \varepsilon).$
For 
$\nu = \sqrt{\sum_\ell b^2_\ell} + b_M \sqrt{\log t} = O(b_M \sqrt{\log t})$, applying Lemma~\ref{lem:sum_of_lap} says that the error from the noise is $O(\nu \sqrt{\log(1/\beta)}) = O(\sqrt{\log(1/\beta)}\log(t)^{\max(0.5, 1.5-\lambda)})$ with probability $1-\beta$,  proving Theorem~\ref{thm:detailed_main_upper}.

\paragraph{Proof of Corollary~\ref{cor:main}.}
For releasing $T$ outputs, choosing $\beta=1/T^{c+1}$, $c>0$ being a constant, and using a union bound over all outputs gives a bound on the maximum noise equal to $O(\log(T)^{\max(1,2-\lambda)} / \varepsilon)$ with probability $1-1/T^c$, proving Corollary~\ref{cor:main}.
\section{Lower Bound on the Privacy Decay}%
\label{sec:lowerbound}
The lower bound follows from a careful packing argument. The proof is deferred to Appendix~\ref{sec:omitted_lowerbound}.
\begin{theorem}\label{thm:lower}
    Let $\Alg$ be an algorithm for binary counting for streams of length $T$ 
    which satisfies $\epsilon$-differential privacy with expiration $g$. Let $C$ be an integer and assume that the additive error of $\Alg$ is bounded by $C<T/2$ at all time steps with a probability of at least 2/3.  Then     %
    \begin{align*}
        \sum_{j=0}^{2C-1}g(j) \geq \log(T/6C)/\epsilon
    \end{align*}
\end{theorem}

The lower bound extends to mechanisms running on unbounded streams.
By Definition~\ref{def:decayed}, $g(j)$ is non-decreasing in $j$. This immediately gives \Cref{thm:main_lower}.
\section{Empirical Evaluation}\label{sec:empirical}
We empirically evaluated (i) how the privacy loss increases as the elapsed time increases for Algorithm~\ref{alg:privacy_degradation}, (ii) how tightly the corresponding theoretical expiration function $g$ of Theorem~\ref{thm:detailed_main_upper} bounds this privacy loss, and (iii) how this privacy loss compares to the privacy loss of a realistic baseline.
As different algorithms have different parameters that can affect privacy loss, we use the following approach to perform a fair comparison:
In the design of $\epsilon$-differentially private algorithms the error of different algorithms is frequently measured with the same value of the privacy loss parameter $\epsilon$. Here, we turn this approach around: {\em  We compare the privacy loss (as a function of elapsed time) of different algorithms whose privacy parameter $\epsilon$ is chosen to achieve the same error.}

We empirically compute the privacy loss for Algorithm~\ref{alg:privacy_degradation} by considering the exact dyadic decompositions used for the privacy argument (see Section~\ref{sec:empirical_details} for details).
As a baseline to compare against, we 
break the input stream into intervals of length $W$, run the 'standard' binary mechanism $\mathcal{A}_B$ of Chan et al.~\citet{chan2011private} with a privacy parameter $\epsilon_{cur}$ on the current interval, and compute the sum of all prior intervals with a different privacy guarantee $\epsilon_{past}$.
As for both algorithms that we evaluate it is straightforward to compute the \emph{mean-squared error} (MSE) for all outputs on a stream of length $T$, while the corresponding maximum absolute error can only be observed empirically, we fix the MSE for $T=d_{max} + 1$, where $d_{max}$ is the greatest $d$ (on the $x$-axis) shown in each plot.
We normalize each plot to achieve the same MSE over the first $T$ outputs, across all algorithms and parameter choices.
For all runs of Algorithm~\ref{alg:privacy_degradation} we used $B=0$, as for larger values of $B$, the primary effect would be to shift the privacy loss curve to the right.
We picked the MSE to be $1000$ for all plots as it leads to small values of the empirical privacy loss.
As for both algorithms, the privacy loss {\em does not} depend on the input data; we used an all-zero input stream, a standard approach in the industry (see, for example, Thakurta's~\citet{thakurta2017differential} plenary talk at USENIX, 2017).
\begin{figure*}[htb]
\centering
\subfigure[Alg.~\ref{alg:privacy_degradation} vs. $g$ in Theorem~\ref{thm:main_upper}.]{\includegraphics{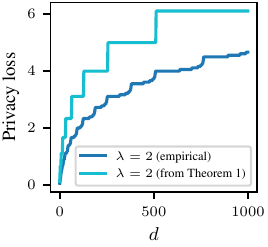}\label{fig:main_algorithm1_empirical_vs_theorem}}
\subfigure[Alg.~\ref{alg:privacy_degradation} and baseline for various $\lambda$, resp.,~$W$.]{\includegraphics{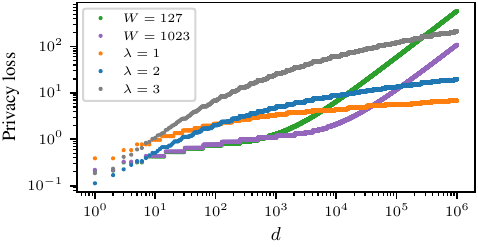}\label{fig:main_plot}}
\caption{Plots on the privacy loss for our \Cref{alg:privacy_degradation} and a baseline algorithm.}
\end{figure*}
\Cref{fig:main_algorithm1_empirical_vs_theorem} shows that $g$ from Theorem~\ref{thm:main_upper} is a good approximation of the empirical privacy loss, and that both exhibit the same polylogarithmic growth. 
\Cref{fig:main_plot} shows that for large enough $d$ our algorithm has lower privacy loss than the baseline algorithm.
See more details in the appendix.

\section{Conclusion}\label{sec:conclusion}

In this work, we give the first algorithm for the continual counting problem for privacy with expiration for a wide range of expiration functions and characterize for which expiration functions it is possible to get an $\ell_{\infty}$-error of $O((\log T)/\epsilon)$. We also give a general lower bound for any such algorithms and show that our is tight for certain expiration functions.
It would be interesting to study such an algorithm further, i.e., slower-growing expiration functions, and also algorithms for other problems in continual observation, such as maintaining histograms and frequency-based statics over changing data.
One of the main applications of continual counting algorithms is in privacy-preserving federated learning algorithms, specifically in stochastic gradient descent (see e.g.~\cite{denisov2022improved,kairouz2021practical, mcmahan2022federated}).
It would be interesting to explore how our algorithm can be deployed in this setting.

Though the concept of privacy expiration has not been defined for approximate differential privacy, it is natural to wonder if there exist analogous results in this setting, which, in general, allows better privacy-utility trade-offs.
We note that for $\rho$-zero-concentrated differential privacy~\cite{bun2016concentrated} there is a natural analog of Definition~\ref{def:decayed} for which it seems possible to prove results analogous to those shown here for pure differential privacy.

\section{Acknowledgements}
\erclogowrapped{5\baselineskip}Monika Henzinger:  This project has received funding from the European Research Council (ERC) under the European Union's Horizon 2020 research and innovation programme (Grant agreement No.\ 101019564) and the Austrian Science Fund (FWF) grant DOI 10.55776/Z422, grant DOI 10.55776/I5982, and grant DOI 10.55776/P33775 with additional funding from the netidee SCIENCE Stiftung, 2020–2024.

Joel Daniel Andersson and Rasmus Pagh are affiliated with Basic Algorithms Research Copenhagen (BARC), supported by the VILLUM Foundation grant 16582, and are also supported by Providentia, a Data Science Distinguished Investigator grant from Novo Nordisk Fonden.

Jalaj Upadhyay's research was funded by the Rutgers Decanal Grant no. 302918. 

Teresa Steiner is supported by a research grant (VIL51463) from
VILLUM FONDEN.

\bibliography{main}
\bibliographystyle{plain}

\newpage
\appendix

\section{Empirical Evaluation}\label{sec:empirical}
In this section, we empirically review (i) how the privacy loss increases as the elapsed time increases for Algorithm~\ref{alg:privacy_degradation}, (ii) how tightly the corresponding theoretical expiration function $g$ of Theorem~\ref{thm:detailed_main_upper} bounds this privacy loss, and (iii) how this privacy loss compares to the privacy loss of a realistic baseline.
Note that different algorithms have different parameters that can affect privacy loss; thus, it is not clear how to perform a fair comparison at first.
We use the following approach: In the design of $\epsilon$-differentially private algorithms the error is frequently bound as a function of the privacy loss, i.e., the error of different algorithms is measured with the same value of the privacy loss parameter $\epsilon$. Here, we turn this approach around: {\em  We compare the privacy loss (as a function of elapsed time) of different algorithms whose privacy parameter $\epsilon$ is chosen to achieve the same error.}
As for both algorithms that we evaluate it is straightforward to compute the \emph{mean-squared error} (MSE) for all outputs on a stream of length $T$, while the corresponding maximum absolute error can only be observed empirically, we fix the MSE for $T=d_{max} + 1$, where $d_{max}$ is the greatest $d$ shown in each plot.
We normalize each plot to achieve the same MSE over the first $T$ outputs, across all algorithms and parameter choices.
For all runs of Algorithm~\ref{alg:privacy_degradation} we used $B=0$, as for larger but practical values of $B$, the primary effect would be to shift the privacy loss curve to the right.
We picked the MSE to be $1000$ for all plots as it leads to small values of the empirical privacy loss.
Any other choice would simply rescale the values of the $y$-axes.

Note that for both algorithms the privacy loss {\em does not} depend on the input data and, thus, we used an all-zero input stream, a standard approach in the industry (see, for example, ~\citet{thakurta2017differential}'s plenary talk at USENIX, 2017).
We also emphasize that all computations producing the plots shown are deterministic, and so there is no need for error bars.

\textbf{Empirical privacy expiration for Algorithm~\ref{alg:privacy_degradation}.}
For our evaluation we empirically compute the privacy loss for Algorithm~\ref{alg:privacy_degradation} by considering the exact dyadic decompositions used for the privacy argument.
More specifically, to compute the empirical privacy expiration function $g$ (for $d \geq B$) of Algorithm~\ref{alg:privacy_degradation}, we use the  proof of Theorem~\ref{thm:main_upper} to reason that for $\mathcal{N}\in \support(Z)$, $\Pr_{z\in Z}[z\in\mathcal{N}]$ is bounded by 
\cref{eq:privacyloss}. There
we bound the exponent, i.e.\,the privacy loss, as a function of the size of the interval $d - B + 1 =\tau' - j + 1$, and this yields the $g$ stated. 
However, this might not be exact, as the decomposition of an interval of size $d - B + 1$ does not necessarily involve using $2$ intervals per level up to $\ell^*=\lfloor\log(d - B + 1)\rfloor$.
Instead, for a given $d \geq B$ and $\lambda$, we can for all $\tau' \geq j$, where $\tau' - j + 1 = d - B + 1$, compute the exact associated dyadic decomposition $D_{[j, \tau']}$ and the resulting privacy loss.
For each value of $d$ taking the maximum over these losses over all $\tau' \le t$ gives the worst-case privacy loss at time $t+d$ for an input streamed at time $t$.
\begin{figure*}[htb]
\centering
\subfigure[Alg.~\ref{alg:privacy_degradation} vs. $g$ in Theorem~\ref{thm:main_upper}.]{\includegraphics[width=0.3\columnwidth]{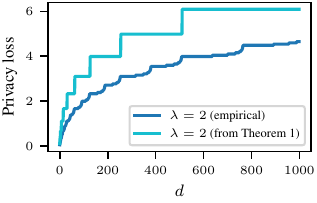}\label{fig:algorithm1_empirical_vs_theorem}}
\subfigure[Algorithm~\ref{alg:privacy_degradation} for multiple $\lambda$.]{\includegraphics[width=0.3\columnwidth]{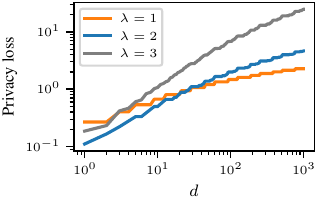}\label{fig:algorithm1_multiple_lambda}}
\subfigure[Baseline for multiple $W$.]{\includegraphics[width=0.3\columnwidth]{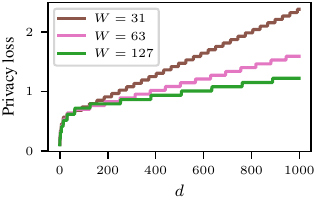}\label{fig:baseline}}
\caption{
Worst-case privacy loss computed empirically for a data item streamed $d$ steps earlier.
}
\end{figure*}

Figure~\ref{fig:algorithm1_empirical_vs_theorem} demonstrates how the privacy loss compares to what is predicted by Theorem~\ref{thm:main_upper}.
The main observation is that $g$ from Theorem~\ref{thm:main_upper} is a good approximation of the empirical privacy loss, and that both exhibit the same polylogarithmic growth.

Figure~\ref{fig:algorithm1_multiple_lambda} plots the empirically computed privacy loss for different choices of $\lambda$.
It demonstrates that the choice of $\lambda$ equates to a trade-off between short-term and long-term privacy. The larger $\lambda$ the higher the privacy loss for large values of $d$, which is to be expected. It also shows that a smaller $\lambda$ provides greater privacy loss early on.

\textbf{A baseline.}
As a baseline to compare against, we 
break the input stream into intervals of length $W$, run the 'standard' binary mechanism $\mathcal{A}_B$ of~\citet{chan2011private} with a privacy parameter $\epsilon_{cur}$ on the current interval, and compute the sum of all prior intervals with a different privacy guarantee $\epsilon_{past}$.
More exactly, the baseline outputs $\mathcal{A}_B(\epsilon_{cur}, t)$ for all time steps $t$ in the first round, i.e., $t \in [1,W]$. For each later round $r>1$, it does the following:
        (1) Compute $c_r = z_{past} + \sum_{t=1}^{(r-1)\cdot W} x_t$ where $z_{past}\sim\Lap(1/\epsilon_{past})$.
        (2) For $t\in [(r-1)\cdot W+1, r\cdot W]$, output $c_r + \mathcal{A}_B(\epsilon_{cur}, t)$ 

Intuitively, this models what is currently done in the software industry: 
Each user is given a ``privacy budget'' for a fixed time interval with the guarantee that their data is no longer used whenever the budget reaches 0 {\em within the current interval}. However, at the beginning of the next interval, the privacy budget is being reset. Thus the total privacy loss is the sum of the privacy losses in all intervals.

What fraction of the privacy budget is spent on the binary mechanism ($\epsilon_{cur}$) and what is spent on releasing the prefix of the past ($\epsilon_{past}$) can be chosen in multiple ways.
For our experiments, we choose a fixed fraction of $\epsilon_{past}/\epsilon_{cur} = 0.1$, implying that we release sums from past rounds with a stronger privacy guarantee compared to what is used in the binary tree.
Other choices and their implication are discussed in Appendix~\ref{sec:empirical_details}.

Figure~\ref{fig:baseline} shows the privacy loss of the baseline for a selection of round sizes.
The main feature to underline is that, after $W$ time steps, the privacy expiration becomes linear in the number of rounds (and therefore in $d$).
This is a direct consequence of privacy composition:
An input $x_t$ impacts the outputs of $\mathcal{A}_B$ in the round $r=\lceil t / W\rceil$ it participates in, and then is subsequently released in each future round $r' > r$ as part of $c_{r'}$.

\textbf{Comparing Algorithm~\ref{alg:privacy_degradation} to the baseline.}
Algorithm~\ref{alg:privacy_degradation} is compared to the baseline in Figure~\ref{fig:comparison}.
Qualitatively the results align with our expectations.
For large enough $d > W$, the baseline enters the region where the degradation in privacy is dominated by the number of rounds that a given input participates in, yielding linear expiration.
Notably, the baseline generally achieves a lower privacy loss for smaller $d$ compared to our method.
This is largely decided by how $\epsilon_{past}/\epsilon_{cur}$ is picked.
Choosing a smaller fraction would lower the slope in the linear regime at the expense of the early privacy loss incurred from $\mathcal{A}_B$.
By contrast, Algorithm~\ref{alg:privacy_degradation} exhibits a comparable initial privacy loss that grows more slowly -- never reaching linear privacy loss -- as predicted by our theorem.
\begin{figure}[ht]
\centering
\includegraphics{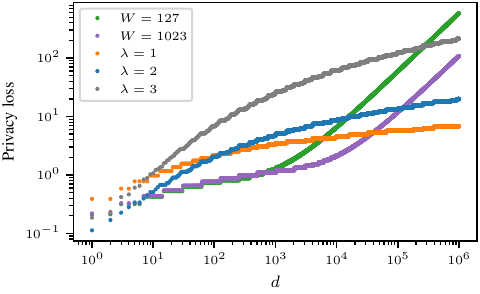}
\caption{
Worst-case privacy loss for a data item streamed $d$ steps earlier, shown for Algorithm~\ref{alg:privacy_degradation}  (with $\lambda=1, 2, 3$) versus the baseline ($W=127$ and $W=1023$).
}\label{fig:comparison}
\end{figure}

\subsection{Technical Details}\label{sec:empirical_details}

\textbf{Choosing $\epsilon_{past}/\epsilon_{cur}$ for the baseline}
For our experiments above, we choose a fraction $\epsilon_{past}/\epsilon_{cur} = 0.1$ for our baseline, but, as we state, other choices are possible.
In particular, one appealing choice is the fraction that minimizes the privacy loss at $d=T-1$, making the baseline as privacy preserving as possible for the last $d$ shown in Figures~\ref{fig:baseline}~and~\ref{fig:comparison}.
Note that the privacy loss for the largest value of $d$ will be, roughly, $\epsilon_{cur} + \epsilon_{past}\cdot (N-1)$, where $N$ is the total number of rounds, corresponding to the privacy loss for an input participating in the first round.
We can compute the fraction $\epsilon_{past}/\epsilon_{cur}$ that minimizes this privacy loss under the constraint of having a fix mean-squared error, yielding a solution that is a function of the round length $W$ and input stream length $T$.

The Figures~\ref{fig:baseline}~and~\ref{fig:comparison} from Section~\ref{sec:empirical} are shown below as Figures~\ref{fig:baseline_opt}~and~\ref{fig:comparison_baseline_opt} where we instead of using $\epsilon_{past}/\epsilon_{cur}=0.1$, we use the fraction minimizing the maximum privacy loss.
\begin{figure}[ht]
\centering
\subfigure[Baseline for multiple $W$.]{\includegraphics[width=0.49\columnwidth]{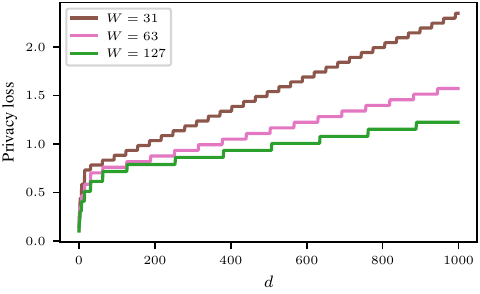}\label{fig:baseline_opt}}
\subfigure[Algorithm~\ref{alg:privacy_degradation}. vs. baseline]{\includegraphics[width=0.49\columnwidth]{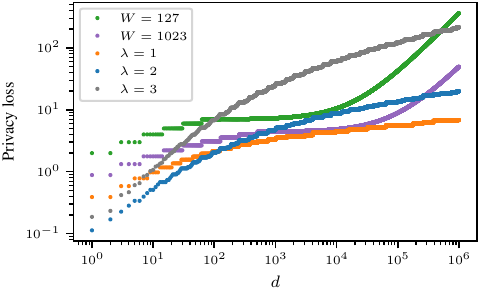}\label{fig:comparison_baseline_opt}}
\caption{
Worst-case privacy loss computed empirically for a data item streamed $d$ steps earlier.
Figure~\ref{fig:baseline_opt} is a re-computation of Figure~\ref{fig:baseline} where the ratio $\epsilon_{past}/\epsilon_{cur}$ is set to minimize the maximum privacy loss, yielding a ratio of $0.069$ for $W=31$, $0.08$ for $W=63$ and $0.095$ for $W=127$.
Figure~\ref{fig:comparison_baseline_opt} is a re-computation of Figure~\ref{fig:comparison} where the ratio $\epsilon_{past}/\epsilon_{cur}$ is set to minimize the maximum privacy loss, yielding a ratio of $0.0064$ for $W=127$ and $0.010$ for $W=1023$.
}
\end{figure}
Note that Figure~\ref{fig:baseline_opt} is almost identical to Figure~\ref{fig:baseline}.
This is due to the fact that for this choice of $W$ and $T$, the optimal fraction is close $0.1$.
In the case of Figure~\ref{fig:comparison_baseline_opt} however, there is a notable difference.
The minimizing fraction is here considerably smaller, resulting in a tangible reduction of the final privacy loss, but at the expense of the privacy loss early on.
This is in line with expectations: spending more of the privacy budget on releasing sums from past rounds implies spending less on releasing the binary tree in each round, which dominates the privacy loss for small $d$.

While the fraction $\epsilon_{past}/\epsilon_{cur}$ has an impact, running the baseline for a fix round length $W$ will always result in linear privacy expiration for a great enough stream length $T$, where $\epsilon_{past}$ affects the slope.

\subsection{Exact Parameters for the Experiments}

In all the plots, we choose the privacy parameter(s) to achieve a mean-squared error of $1000$ over the stream length $T$ in each figure shown, where $T=10^3$ for Figure~\ref{fig:algorithm1_empirical_vs_theorem}, \ref{fig:algorithm1_multiple_lambda} and \ref{fig:baseline}, and $T=10^6$ in Figure~\ref{fig:comparison}.
The corresponding privacy parameters are listed in Table~\ref{tab:experiment_epsilons}, $\epsilon$ refers to the privacy parameter used by Algorithm~\ref{alg:privacy_degradation} and $\epsilon_{cur}$, $\epsilon_{past}$ are the ones used by the baseline.
\begin{table}[htbp]
\centering
\caption{Table over the privacy parameters used in each of the plots.}
\small{\begin{tabular}{ccccc}
Figure & Algorithm & $\epsilon$ &  $\epsilon_{cur}$ & $\epsilon_{past}$ \\ \toprule
\ref{fig:algorithm1_empirical_vs_theorem} & $\lambda=2$ & 0.05542 & - & - \\ \midrule
\multirow{3}{*}{\ref{fig:algorithm1_multiple_lambda}} & $\lambda=1$ & 0.1341 & - & - \\
& $\lambda=2$ & 0.05542 & - & - \\
& $\lambda=3$ & 0.04651 & - & - \\ \midrule
\multirow{3}{*}{\ref{fig:baseline}} & $W=31$ & - & 0.5678 & 0.05678 \\
& $W=63$ & - & 0.6372 & 0.06372 \\
& $W=127$ & - & 0.7197 & 0.07197 \\ \midrule
\multirow{5}{*}{\ref{fig:comparison}} & $W=127$ & - & 0.7387 & 0.07387 \\
& $W=1023$ & - & 1.096 & 0.1096 \\ 
& $\lambda=1$ & 0.1947 & - & - \\
& $\lambda=2$ & 0.05645 & - & - \\
& $\lambda=3$ & 0.04652 & - & - \\ \midrule
\multirow{3}{*}{\ref{fig:baseline_opt}} & $W=31$ & - & 0.7328 & 0.05048 \\
& $W=63$ & - & 0.7031 & 0.05796 \\
& $W=127$ & - & 0.7170 & 0.07252 \\ \midrule
\multirow{5}{*}{\ref{fig:comparison_baseline_opt}} & $W=127$ & - & 6.973 & 0.04488 \\
& $W=1023$ & - & 4.413 & 0.04589 \\ 
& $\lambda=1$ & 0.1947 & - & - \\
& $\lambda=2$ & 0.05645 & - & - \\
& $\lambda=3$ & 0.04652 & - & - \\ \midrule
\end{tabular}
}
\label{tab:experiment_epsilons}
\end{table}

\section{Background on Random Variables}
\begin{definition}[Laplace Distribution]
The \emph{Laplace distribution} centered at $0$ with scale $b$ is the distribution with probability density function %
\begin{align*}
    \densLap{b}(x)=\frac{1}{2b}\exp\left(\frac{-|x|}{b}\right).
\end{align*}
 We use $X\sim \Lap(b)$ or just $\Lap(b)$ to denote a random variable $X$ distributed according to $\densLap{b}(x)$.
\end{definition}

\begin{fact}[Laplace Tailbound]\label{fact:laplace_tailbound}
If $X\sim \Lap(b)$, then 
\begin{align*}
    \Pr[|X| > t\cdot b]\leq e^{-t}.
\end{align*}
\end{fact}

\begin{lemma}[Measure Concentration Lemma in \citet{chan2011private}]\label{lem:sum_of_lap}
Let $Y_1,\dots,Y_k$ be independent variables with $Y_i\sim\Lap(b_i)$ for all $i\in[k]$. Denote $b_M=\max_{i\in[k]}b_i$ and $Y=\sum_{i=1}^k Y_i$. Let $0<\beta<1$ and $\nu>\max\left\{\sqrt{\sum_{i=1}^k b_i^2},b_M\sqrt{\ln(2/\beta)}\right\}$. Then,
\begin{align*}
    \Pr\left[|Y|>\nu\sqrt{8\ln(2/\beta)} \right] \leq \beta.
\end{align*}
\end{lemma}

\begin{corollary}\label{cor:sum_of_eq_lap}
Let $Y_1,\dots,Y_k$ be independent variables with distribution $\Lap(b)$ and let $Y=\sum_{i=1}^k Y_i$. Let $0<\beta<1$. Then
\begin{align*}  \Pr\left[|Y|>2b\sqrt{2\ln(2/\beta)}\max\left\{\sqrt{k},\sqrt{\ln(2/\beta
    )}\right\} \right] \leq \beta.
\end{align*}
\end{corollary}
\begin{proof}
Apply Lemma~\ref{lem:sum_of_lap} to $b_1=\dots=b_k=b$.
\end{proof}

\section{Omitted proofs}\label{sec:omitted_proofs}
\subsection{Omitted proofs from \Cref{sec:prelims}}
\label{sec:omitted_proofs_prelims}
\begin{proof}[Proof of Fact~\ref{fact:varialbe_shift}]
Given $x$, $x'$ and any set of output sequences set $S$, let $\mathcal{N}_S$ be the set of all choices of random variables $z$ such that $\Alg(x ;z)\in S$. Note that for all $z \in \mathcal{N}_S$, $\Alg(x'; q(z)) \in S$, i.e., for all $z' \in q(\mathcal{N}_S)$, $\Alg(x'; z') \in S$. Then 
\begin{align*}
    \Pr_{z\sim Z}[\Alg(x; z)\in S] &=\Pr_{z\sim Z}[z\in \mathcal{N}_S] \leq e^{\epsilon g(\tau-j)}\,\Pr_{z'\sim Z}[z'\in q(\mathcal{N}_S)] \\
    &\leq e^{\epsilon g(\tau-j)}\, \Pr_{z'\sim Z}[\Alg(x'; z')\in S].
\end{align*}
This completes the proof of \Cref{fact:varialbe_shift}.
\end{proof}

\begin{proof}[Proof of Lemma~\ref{lem:lap_shift}]
    Let $\densLap{b}$ denote the density function of $\Lap(b)$. For all $x,y\in\mathbb{R}$:
     \begin{align*}\densLap{b}(x)&=\frac{1}{2b}e^{-|x|/b} \leq \frac{1}{2b}e^{(|y|-|x+y|)/b} = e^{|y|/b}\frac{1}{2b}e^{(-|x+y|)/b} = e^{|y|/b}\densLap{b}(x+y),
    \end{align*}
    where the inequality follows from the triangle inequality $|x+y|\leq |x|+|y|$.
    In the following integral%
    , let $z'=q(z)$. Let $s$ be as defined in the lemma statement. Then we have
    \begin{align*}
        \Pr_{z\in Z}[z\in\mathcal{N}]&=\int\limits_{z\in\mathcal{N}}\prod_{i=1}^k \densLap{b_i}(z_i)dz\leq \int\limits_{z\in\mathcal{N}}\prod_{i=1}^k e^{|z'_i-z_i|/b_i} \densLap{b_i}(z'_i)dz\\
        & = e^s\int\limits_{z\in\mathcal{N}}\prod_{i=1}^k \densLap{b_i}(z'_i)dz = e^s\int\limits_{z\in q(\mathcal{N})}\prod_{i=1}^k \densLap{b_i}(z_i)dz\\
        &= e^s\Pr_{z\sim Z}[z\in q(\mathcal{N})]
    \end{align*}%
This completes the proof of \Cref{lem:lap_shift}.
\end{proof}

\subsection{Omitted proofs from \Cref{sec:warmup}}
\label{sec:omitted_warm_up_proof}
\begin{proof}[Proof of Lemma~\ref{lem:simple}]
{\bf Privacy.}
We use Fact~\ref{fact:varialbe_shift} and Lemma~\ref{lem:lap_shift} to argue the privacy of our simple algorithm, $\Alg_{\mathsf{simple}}$: Let $x$ and $x'$ differ at time $j$. 
Note that the prefix sums fulfill the following properties:
\begin{itemize}
   \item $\sum_{i=0}^t x_i=\sum_{i=0}^t x'_i$ for all $t<j$ and
   \item $\sum_{i=0}^t x'_i = \sum_{i=0}^t x_i+y$ for all $t \geq j$, where $y=x'_j-x_j\in[-1,1]$.
\end{itemize}
Let $S\subseteq\range(\Alg_{\mathsf{simple}})^*$. For any $\tau \geq j$, we consider the output distribution of length-$\tau$ prefixes. There are two cases: 
For $j=\tau$ and any $t\leq \tau$, the output on $x$ is given by $\sum_{i=1}^{t-1}x_i+Z_{t-1}$ with $t-1<j$, thus we get equal output distributions on $x$ and $x'$. That is, \begin{align*}\Pr[\Alg(x_1,\dots, x_{t})_{t=1}^{\tau}\in S]=\Pr[\Alg(x'_1,\dots,x'_t)_{t=1}^{\tau}\in S] =e^{g(0)\epsilon}\Pr[\Alg(x'_1,\dots,x'_t)_{t=1}^{\tau}\in S].\end{align*}

For $j<\tau$, let $Z=\randvar_1,\dots,\randvar_{\tau-1}$ be the sequence of Laplace random variables used by the algorithm.  For any fixed output sequence $s\in S$, note that the algorithm guarantees that $s_1 = 0$ and let $z=z_1,z_2,\dots,z_{\tau-1}$ be the values that the Laplace random variables need to assume to get output sequence $s$ with input $x$. That is, $\sum_{i=1}^{t} x_i + z_t=s_{t+1}$ for $t\geq 1$. We define a bijection $q$ satisfying the properties of Fact~\ref{fact:varialbe_shift} as follows. 
Define $q(z)=z'$ such that 
\begin{align*}
    z'_i = 
    \begin{cases}
    z_i & i \leq j \\
    z_i + y & i >j
    \end{cases}. 
\end{align*}

 This gives $\sum_{i=1}^{t} x'_i + z'_t=s_{t+1}$. All Laplace noises have the same distribution $\Lap(1/\epsilon)$. By Lemma~\ref{lem:lap_shift}, we have for any $\mathcal{N}\in\support(Z)$,  $$\Pr_{z\in Z}[z\in\mathcal{N}]\leq e^{(\tau-j)\epsilon}\Pr_{z\in Z}[z\in q(\mathcal{N})].$$ 

Thus, $q$ fulfills the properties of Fact~\ref{fact:varialbe_shift}, and therefore our algorithm satisfies $\epsilon$-differential privacy with expiration $g(t)=t$.

{\bf Accuracy.}
The error at time step $t$ is given by $x_t+\randvar_{t-1}$. By the Laplace tail bound (Fact~\ref{fact:laplace_tailbound}), we have that $$\Pr_{z_i\sim \Lap(1/\epsilon)}[|z_{i}| > \epsilon^{-1}\log(T/\beta)]\leq \beta/T$$ for any $i\in[T-1]$. 

Therefore, by a union bound, $|\randvar_{i}|\leq \log(T/\beta)$ simultaneously for all $i\in[T-1]$ with probability at least $1-\beta$. This implies that, with probability at least $1-\beta$, the maximum error over the entire stream is bounded by $\epsilon^{-1}\log(T/\beta)$. 
\end{proof}

\begin{proof}[Proof sketch of Fact~\ref{fact:dyadic_decomposition}]
For $a=b$, the claim is immediate and so we prove it for $a < b$.
First, consider the case where $a=2^{\ell}$ and $b\leq 2\cdot 2^{\ell} - 1$.
We show that in this case, there exists a set $D_{[a,b]}$ with the properties above, only it contains at most \emph{one} interval per level.
If $b=2\cdot 2^{\ell} - 1$, then $D_{[a,b]}=\{[a,b]\}$.
Else, we have $b-a+1<2^{\ell}$, thus there exists a binary representation $(q_0,\dots, q_{\ell-1})\in \{0,1\}^{\ell}$ such that $b-a+1=q_0+2q_1+2^2q_2+\dots+ 2^{\ell-1}q_{\ell-1}$.
Note that for any $j\leq \ell$, there is an interval in $\mathcal{I}$ of level $j$ starting at $a=2^{\ell}$.
We now show how to construct the set $D_{[a,b]}$.
In the first step, choose the largest $j$ such that $q_j=1$ and add the interval $[2^{\ell},2^{\ell}+2^j - 1]$ to $D_{[a,b]}$.
Then, set $q_j=0$.
Next, assume we already cover the interval $[a,s]$ for some $b\geq s>a$.
We then again pick the largest remaining $j$ such that $q_j=1$ and add the interval $[s,s+2^j - 1]$ to $D_{[a,b]}$.
By an inductive argument, this interval will be in $\mathcal{I}$, since $s$ is the ending position of an interval of a higher level in $\mathcal{I}$.
We do this until no more $q_j$ with $q_j=1$ remains, at which point we have covered $[a,b]$.
The same argument can be repeated for the case where $b=2^{\ell}-1$ and $a\geq 1$.
Now, for the general case, pick $\ell=\lfloor \log(b)\rfloor$ and $c=2^{\ell} - 1$ and define $D_{[a,b]}=D_{[a,c]}\cup D_{[c+1,b]}$. 
\end{proof}

\begin{proof}[Proof sketch of Fact~\ref{fact:dyadic_levels}]
First note that in our definition of $\mathcal{I}$, we exclude all intervals starting at $0$.
For a given $t$, let $\ell\in\mathbb{N}$ be the greatest interval level satisfying $2^\ell-1 < t$.
It follows that for every $j > \ell$, $t \in [0, 2^j - 1] \notin \mathcal{I}$.
Analogously, for every $j \leq \ell$ we have $t\notin [0, 2^j - 1]$, and therefore there must exist a unique $I\in\mathcal{I}^j$ such that $t\in I$.
Taken together, we have that $|\mathcal{I}_t| = \ell + 1$ where $\ell = \lfloor \log t \rfloor$.
\end{proof}

\subsection{Omitted Proofs from Section~\ref{sec:main_proof}}
\label{sec:missingproof}
In this section, we give the detail proof for privacy analysis for all $\lambda \in \mathbb R_{\geq 0}$ and accuracy analysis for $\lambda=3/2$, which is not covered in \Cref{sec:main_proof}. 

\subsubsection{Privacy proof for all $\lambda \in \mathbb R_{\geq 0}$}

\begin{proof}[Proof of Lemma~\ref{lem:mainalgdecay}]

We prove a slightly more general result; i.e., for all $\lambda \in \mathbb R_{\geq0}$, we have the following bound on the privacy expiration function:
\[
g(d) \leq \begin{cases}
    0 & \text{for } d < B \\
    2 \left(1 + {\bigg[\frac{1}{\lambda} \left((\log(d - B + 1) + 1)^{\lambda} - 1 \right) \bigg]}\right) & \text{for } B \leq d \text{ and } \lambda \neq 0 \\
    2\left(1 + \log\log(d-B+1)\right) & \text{for } B \leq d \text{ and } \lambda = 0
\end{cases}
\]
We wish to use \Cref{fact:varialbe_shift}. For that, we first show that both the requirements in \Cref{fact:varialbe_shift} are satisfied.

\begin{enumerate}
    \item By definition, $\Alg(x_1 x_2\dots x_t; z) = \sum_{i=1}^t x_i + \sum_{I\in \mathcal{I}_t} z_I$.
    For $t< j$, we have $t\notin [j,\tau']$  and therefore $t\notin I$ for all $I\in D_{[j,\tau']}$. It follows that for all $I \in \mathcal{I}_t$ it holds that $I \not\in D_{[j,\tau']}$.
    Therefore, we have 
    $$\sum_{i=1}^t x_i+\sum_{I\in \mathcal{I}_t} z_I=\sum_{i=1}^t x'_i+\sum_{I\in \mathcal{I}_t} z_I=\sum_{i=1}^t x'_i+\sum_{I\in \mathcal{I}_t} z'_I.$$
    
    For $t\geq j$, we have that $t$ is contained in exactly one $I\in D_{[j,\tau']}$.%
    Thus, $z'_I=z_I+y$ for exactly one $I\in\mathcal{I}_t$, and $z'_{I'}=z_{I'}$ for all $I'\in \mathcal{I}_t\backslash\{I\}$. 
    Further, since $t\geq j$, we have that $\sum_{i=1}^t x_i=\sum_{i=1}^t x'_i-y$. 
    Together, we have 
    \begin{align*}
        \sum_{i=1}^t x_i+\sum_{I\in \mathcal{I}_k} z_I &=\sum_{i=1}^t x'_i-y+\sum_{I\in \mathcal{I}_t} z'_I+y =\sum_{i=1}^t x'_i+\sum_{I\in \mathcal{I}_t} z'_I,
    \end{align*} 
    so the first property of Fact~\ref{fact:varialbe_shift} is fulfilled.

    \item By Fact~\ref{fact:dyadic_decomposition}, $D_{[j,\tau']}$ consists of at most two intervals from each of $\mathcal{I}^0,\dots,\mathcal{I}^{\ell^*}$ where $\ell^* = \lfloor \log(\tau' - j + 1)\rfloor$.
    Thus, by Lemma~\ref{lem:lap_shift} we have for any $\mathcal{N}\in \support(Z)$, then 
    \begin{align}
    \begin{split}
        \Pr_{z\in Z}[z\in\mathcal{N}] & \leq \exp\left(\sum_{\ell=0}^{\log T} \sum_{I\in D_{[j,\tau']}\cap \mathcal{I}^\ell} {|y| \epsilon \over  (1+\ell)^{1 - \lambda}}\right)\, \Pr_{z\in Z}[z\in q(\mathcal{N})] \\
        &\leq \exp\left(2 \epsilon \sum_{\ell=0}^{\ell^*} (1+\ell)^{\lambda - 1}\right)\, \Pr_{z\in Z}[z\in q(\mathcal{N})] \enspace .
    \end{split}
    \label{eq:privacyloss}
    \end{align}
\end{enumerate}
Using Fact~\ref{fact:varialbe_shift} with $d = \tau - j = B + \tau' - j$ we conclude differential privacy with privacy expiration $g(d) = 2 \sum_{\ell=0}^{\lfloor\log(d - B + 1)\rfloor} (1+\ell)^{\lambda - 1}$.
Next we bound:
\begin{align*}
    g(d) &= 2  \sum_{\ell=0}^{\lfloor\log(d - B + 1)\rfloor} (1+\ell)^{\lambda - 1}\\
    & \leq 2 \left(1 + \int\limits_{1}^{\log(d - B + 1)+1} x^{\lambda - 1}\, \mathsf dx\right)\\
    &= 2 \left(1 + \bigg[\frac{1}{\lambda} x^{\lambda}\bigg]_{x=1}^{x=\log(d - B + 1) + 1}\right) \\
    &= O( 1 + \log^{\lambda} (d-B+1)) \enspace .
\end{align*}

Note that we use the assumption $\lambda > 0$.

In the case $\lambda = 0$ it instead holds that $g(d) = O(\log\log(d-B+1))$ by setting the value of $\lambda=0$ in the first inequality and then using the bound on the Harmonic sum.

 This is not a point of discontinuity as the $O(\cdot)$ notation would imply. We treated $\lambda$ as a constant and derived the last equation in the above. If instead, we treat $\lambda$ as a parameter, we have the following
\[
g(d) \leq 2 \left(1 + \underbrace{\bigg[\frac{1}{\lambda} \left((\log(d - B + 1) + 1)^{\lambda} - 1 \right) \bigg]}_{f(\lambda)}\right)
\]

We can now evaluate the $f(\lambda)$ as $\lambda \to 0$ as follows: Since $f(\lambda)$ as an indeterminate form and noting that it satisfies the condition required to apply the L'H\^{o}pital rule. Therefore, an application of L'H\^{o}pital rule gives us 
\[
\lim_{\lambda \to 0} f(\lambda) = \lim_{\lambda \to 0} \frac{{\mathsf d \over \mathsf d \lambda}( (\log(d - B + 1) + 1)^{\lambda} - 1)}{{\mathsf d \over \mathsf d \lambda}\lambda} = \lim_{\lambda \to 0} \log(\log(d-B+1))\log^{\lambda}(d-B+1) = \log(\log(d-B+1)).
\]

That is, we have the desired bound on the privacy expiration. This completes the proof of \Cref{lem:mainalgdecay}.

\end{proof}

\subsubsection{Accuracy proof when $\lambda=3/2$}
Recall that in the main text, we stated that for $\lambda=3/2$, we achieve 
\[
\sum_{\ell=0}^{\lfloor\log(t-B)\rfloor}b_\ell^2 \leq 
\sum_{\ell=0}^{\log t} b_\ell^2 =
O\left({\log\log(t)\over \epsilon^2} \right). 
\]
This follows from straightforward computation by setting $\lambda=3/2$:
\[
\sum_{\ell=0}^{\log t} b_\ell^2  = \sum_{\ell=0}^{\log t} {(\ell + 1)^{2(1-\lambda)}\over \epsilon^2 } = \sum_{\ell=0}^{\log t} {1 \over \epsilon^2 (\ell + 1)}= O\left({\log\log(t)\over \epsilon^2} \right).
\]
While it seems like a point of discontinuity based on what we showed in \Cref{sec:main_proof}, we now show that it is not the case. Recall the following computation in \Cref{sec:main_proof} 
\begin{alignat*}{1}
    \sum_{\ell=0}^{\log t} b_\ell^2 & = \sum_{\ell=0}^{\log t} {(\ell + 1)^{2(1-\lambda)}\over \epsilon^2 }\\
    & \leq {1 \over \epsilon^2 } \left(1 + \int\limits_{1}^{\log(t)+2} x^{2(1-\lambda)} \mathsf dx\right)\\
    & \leq {1 \over \epsilon^2 } \left(1 + \bigg[\frac{1}{3 - 2\lambda} x^{3 - 2\lambda}\bigg]_{x=1}^{x=\log(t) + 2} \right)\enspace .
\end{alignat*}
Note that we did not include the last step. Now, setting the limits, we get 
\[
\sum_{\ell=0}^{\log t} b_\ell^2 \leq {1 \over \epsilon^2 } \left(1 + \frac{(\log(t) + 2)^{3 - 2\lambda} - 1}{3 - 2\lambda}  \right)
\]
Just as in the case of the privacy proof, we can compute the limit by applying L'H\^{o}pital rule to get the desired expression.

\subsection{Omitted Proofs from \Cref{sec:lowerbound}}
\label{sec:omitted_lowerbound}

\begin{proof}[Proof of Theorem~\ref{thm:lower}]
    We define $x^{(0)}=0^T$ to be the 0 stream of length $T$.
    Let $T'\leq T$ such that $2C$ divides $T'$, and $T-T'<2C$. Since $T\geq 2C$ we have $T'\geq 2C$ and $T'\geq T/2$. We define $T'/(2C)$ data sets $x^{(1)}, \dots, x^{(T'/(2C))}$ as follows:
    \begin{align*}
        x^{(i)}_t := \begin{cases}
            1 & t\in[2C(i-1)+1,2Ci] \\
            0 & \text{otherwise}
        \end{cases}.
    \end{align*}
    For $i\in[T'/(2C)]$, let $S_i$ be the set of all output sequences $a_1,\dots,a_T\in\range{(\Alg)}$ satisfying $a_{2Ci}>C$, and $a_{2Ck}<C$ for all $k<i$. Note that $S_i\cap S_j=\emptyset$ for $i\neq j$ where $i,j\in[T'/(2C)]$. Further, since by assumption $\Alg$ has error at most $C$ at all times steps with probability at least $2/3$, we have that
    \[\Pr[\Alg(x^{(i)}_1,\dots, x^{(i)}_t)_{t=1}^{2Ci}\in S_i]\geq 2/3.\]
     Recall that $x^{(0)}$ and $x^{(i)}$%
     differ in exactly the positions $t\in[2C(i-1)+1,2Ci]$. By the assumption that $\Alg$ satisfies $\epsilon$-differential privacy with expiration $g$,
    \begin{align*}
        2/3 &\leq \Pr[\Alg(x^{(i)}_1,\dots,x^{(i)}_t)_{t=1}^{2Ci}\in S_i] \\&\leq e^{\sum_{j=0}^{2C-1}g(j)\epsilon}\Pr[\Alg(x^{(0)}_1,\dots,x^{(0)}_t)_{t=1}^{2Ci}\in S_i]\\
        &=e^{\sum_{j=0}^{2C-1}g(j)\epsilon}\Pr[\Alg(x^{(0)}_1,\dots,x^{(0)}_t)_{t=1}^{T'}\in S_i]
    \end{align*}
Now, since the $S_i$'s are disjoint, we have 
\begin{align*}
    1 &\geq \sum_{i=1}^{T'/(2C)}\Pr[\Alg(x^{(0)}_1,\dots,x^{(0)}_t)_{t=1}^{T'}\in S_i]\\&\geq \sum_{i=1}^{T'/(2C)}(2/3)e^{-\sum_{j=0}^{2C-1}g(j)\epsilon} \\
    &=(T'/(2C))(2/3)e^{-\sum_{j=0}^{2C-1}g(j)\epsilon}.
\end{align*}
It follows that 
\[\sum_{j=0}^{2C-1}g(j)\epsilon\geq \log(2T'/(6C))\\\geq \log (T/6C)\]
completing the proof of \Cref{thm:lower}.
\end{proof}

\end{document}